\documentclass[letterpaper, 10 pt, conference, onecolumn]{ieeeconf}  

\IEEEoverridecommandlockouts                              

\usepackage{amsmath,amssymb,amsfonts}
\usepackage{mdframed}
\usepackage{framed}
\usepackage{algorithmic}
\usepackage{graphicx}
\usepackage{textcomp}
\usepackage{thmtools}
\usepackage{mathtools}
\usepackage{comment}

\let\labelindent\relax

\usepackage{enumitem}

\usepackage{booktabs}
\usepackage[utf8]{inputenc}
\usepackage[english]{babel}
\usepackage[dvipsnames]{xcolor}
\usepackage[linesnumbered,lined,boxed,commentsnumbered]{algorithm2e}
\usepackage{xcolor}

\newtheorem{theorem}{Theorem}
\newtheorem{corollary}{Corollary}
\newtheorem{definition}{Definition}

\newcommand{\calk}{\mathfrak{K}}
\newcommand{\TT}{\mathbb{T}}
\newcommand{\KK}{\mathbb{K}}
\newcommand{\FF}{\mathbb{F}}
\newcommand{\RR}{\mathbb{R}}
\newcommand{\NN}{\mathbb{N}}

\newcommand{\EE}{\mathbb{E}}

\newcommand{\knl}{\mathfrak{K}}
\newcommand{\calA}{\mathcal{A}}

\newcommand{\calC}{\mathcal{C}}
\newcommand{\calH}{\mathcal{H}}
\newcommand{\calL}{\mathcal{L}}

\newcommand{\RH}{\mathcal{R}}
\newcommand{\HsubX}{{H}_{{X}}}

\newcommand{\HsubL}{{H}_{L}}
\newcommand{\bbR}{\mathbb{R}}

\newcommand{\EsubXt}{\mathbb{E}_{X_t}}

\newcommand{\ghat}{\hat{g}}

\newcommand{\dghat}{\dot{\hat{g}}}

\newcommand{\gtil}{\tilde{g}}

\newcommand{\Ext}{\mathcal{E}}

\title{Strictly Decentralized Adaptive Estimation of External \\ Fields
using Reproducing Kernels}
\author{Jia Guo$^{*}$, Michael E. Kepler$^{\dagger}$, Sai Tej  Paruchuri$^{*}$, Haoran Wang$^{*}$,\\ Andrew J. Kurdila$^{*}$, and  Daniel J. Stilwell$^{\dagger}$ 
\noindent 
\thanks{\noindent 
    $^{*}$Mechanical Engineering Dept., Virginia Tech, Blacksburg, VA \newline
    $^{\dagger}$Electrical and Computer Engineering Dept., Virginia Tech, Blacksburg, VA 
        }%
}

\begin{document}

\maketitle

\begin{abstract}
This paper describes an adaptive method in continuous time  for the estimation of external fields by a team of $N$ agents.  The agents $i$ each explore subdomains $\Omega^i$ of a bounded subset of  interest $\Omega\subset X:=\RR^d$. Ideal adaptive estimates $\hat{g}^i_t$ are derived for each agent from a  distributed parameter system (DPS) that takes values in the scalar-valued reproducing kernel Hilbert space $H_X$ of functions over $X$.  Approximations of the  evolution of the ideal local estimate $\hat{g}^i_t$ of agent $i$ is constructed solely  using observations made by agent $i$ on a fine time scale. Since the local estimates on the fine time scale  are constructed independently for each agent,  we say that the method is strictly decentralized.  On a coarse time scale, the individual local estimates $\hat{g}^i_t$ are fused via the expression  $\hat{g}_t:=\sum_{i=1}^N\Psi^i \hat{g}^i_t$ that uses a  partition  of unity  $\{\Psi^i\}_{1\leq i\leq N}$ subordinate to the cover $\{\Omega^i\}_{i=1,\ldots,N}$ of $\Omega$. Realizable algorithms are obtained by constructing finite dimensional approximations of the DPS in terms of scattered bases defined by each agent from samples along their trajectories. Rates of convergence of the error in the finite dimensional approximations are derived in terms of the fill distance of the samples that define the scattered centers in each subdomain. The qualitative performance of the convergence rates for the  decentralized estimation method is illustrated via numerical simulations. 
\end{abstract}


\section{Introduction}

\subsection{Overview of the Methodology}
This paper studies the convergence of  the { strictly   decentralized} method for estimation of external fields by agent teams  in the setting of a reproducing kernel Hilbert space (RKHS). This paper is the second in a three part paper that focuses on deriving sharp rates of convergence for adaptive, nonparametric  estimation by teams in an RKHS, where the methodologies vary depending on how information is shared  by the team.    We represent the motion of each agent $i=1,\ldots,N$  by the  trajectory $t\mapsto x^i_t\in X:=\RR^d$. In the estimation task at hand, the agents  collectively explore a bounded subdomain $\Omega\subset X$ of interest, and each individual agent $i$ travels in some assigned   subdomain $\Omega^i\subseteq \Omega$ where $\bigcup_{i=1,\ldots,N}\Omega^i=\Omega$. In this paper we assume that the unknown field that must be estimated is a real-valued function $g:X\rightarrow \RR:=Y$, which  is assumed to reside in an RKHS space $H_X$ of real-valued functions over $X$.  While the agent $i$  traverses the subdomain $\Omega^i$, in principle a continuous family of observations $(x^i_t,y^i_t)\in X\times Y$, the local measurements of agent $i$, are collected. The output $y_t^i:=g(x^i_t)$ is the sample of the unknown field $g$ at the agent's location $x^i_t$ at time $t\in \RR^+$.  In practice, of course, both the state $x^i_t$ and the observation $y^i_t$ are known imprecisely, that is, they are subject to noise. Theoretically, the problem is to devise a decentralized estimation methods in which some subset of the history of local observations $\{(x^i_\tau,y^i_\tau)\}_{\tau\leq t}$ collected by agent $i$ are used, in conjunction with  information shared with some of the other agents, to generate an   estimate by agent $i$  of the unknown function $g$.

Conceptually, the algorithm derived in the paper is defined in two stages. Initially, an ideal estimate of the unknown function $g$ is defined in  terms of a  trajectory of some ideal distributed parameter system (DPS). For the  individual trajectories of the agents, this is denoted $t\mapsto \hat{g}^i(t)$ ,  and for the trajectory of the collective estimate we write  $t\mapsto \hat{g}(t)$. These are  referred to as  ideal estimates since at this point it is only known that they take values in a generally infinite dimensional state space $H_X$. To obtain practical and realizable algorithms, we subsequently define finite dimensional approximations $\hat{g}^i_{L,t}$ and $\hat{g}_{L,t}$, respectively,   of the ideal trajectories $\hat{g}^i_t$ and $\hat{g}_t$. The convergence of the proposed overall method is then bounded, for instance in the study of the error in the collective estimate,  in an inequality such as 
\begin{align*}
   \|g-\hat{g}_{L,t}\|_{\FF} \leq \underset{\text{ideal error}}{\|g-\hat{g}_{t}\|_\FF} + \underset{\text{approximation error}}{\|\hat{g}_{t}-\hat{g}_{L,t} \|_\FF}
\end{align*}
for a suitably defined function space $\FF$. 
We explain in this paper  how to choose the function space $\FF$, define the DPS that determines the ideal estimate,  define the evolution law for the finite dimensional approximations, and how to pick bases to guarantee the convergence of the approximation error. The strongest conclusions of the paper derive rates of convergence for adaptive  estimation that depend on the fill distance of samples in the domain of interest. The authors are not aware of any other such sharp convergence rates for decentralized estimation of functions. 

\subsection{The Centralized Method} 
\label{sec:centralized}
This paper is the second of a three part series where rates of convergence for multi-agent estimation of a function in an RKHS are derived for a variety of information sharing protocols. 
For completeness, we begin with a summary for the {\em centralized} setting. Our new method for strictly decentralized estimation builds on these results. 
The centralized adaptive estimation strategy defines a  collective ideal  estimate $\ghat_t:=\ghat(t,\cdot)\in \HsubX$ that satisfies the equation 
\begin{align}\label{eq: est_eqn}
\dghat_t = \gamma \EsubXt^* (Y_t - \EsubXt \ghat_t) = {\gamma \EsubXt^*\EsubXt} \gtil_t,
\end{align}
where $\EE_{X_t}$ is the vector evaluation operator on  $\HsubX$ evaluated at $X_t:=(x^1_t,\ldots,x^N_t)^T\in X^N$. That is, for any $Z\in X^N$, we have $\EE_Z:H_X\rightarrow \RR^N$ with $$\EE_Z g=(g(z^1),\ldots,g(z^N))^T:=(E_{z^1}g,\ldots,E_{z^N}g)^T \in \RR^N,$$
where for any $z\in X$ the operator $E_z:H_X\rightarrow \RR$ is the evaluation functional $E_zf:=f(z)$.  
 Note that the observations of all the agents are used in the right hand side of the evolution above at each time $t\geq 0$. Forming the equations for  forward propagation of this law requires the state trajectories and observations of all the agents at each time $t$. It is for this reason that we say  that this defines   a centralized  estimation method. The single evolution law above can be implemented by a central processor that collects all the observations of all the agents at each time step.  Practical techniques for building realizable finite dimensional approximations of this ideal estimate is a topic that was  considered in part I of this paper. 

\subsection{The Strictly Decentralized Approach}
In the centralized approach,  the agent collective is just a sensing  network with communication capabilities.  However, in many contemporary situations the members of the  agent teams do have a substantial computing capability.  In fact, it is often desirable to exploit the inherent parallelism afforded by a network that has computational facilities at each node.  For all of these reasons, we describe an  estimation procedure where different propagation rules are used by the different agents.  Each agent constructs a local estimate over some assigned subdomain. Then,  a joint estimate is constructed on a coarser time scale by sharing information about these local estimates. 

 The strictly decentralized approach studied in this paper describes a scenario where, for much of the time, agents work on estimates with no communication from other agents and fuse estimates by sharing information only on occasion or as needed.   The ideal, infinite dimensional estimate of agent $i$ for $i=1,\ldots,N$ is defined as the solution of the equation
\begin{align}
\dot{\hat{g}}^i_t&=\gamma E^*_{x^i_t}\left (y^i_t-E_{x^i_t}\hat{g}^i_t \right ) \notag \\
&=\gamma E^*_{x^i_t}E_{x^i_t} \tilde{g}^i_t
=\calA^i(t)\tilde{g}^i_t,  
    \label{eq:strict1a}
\end{align}
subject to $\hat{g}^i_t|_{t=0}=\hat{g}^i_0$, with $\calA^i(t):=\gamma E^*_{x^i_t}E_{x^i_t}$ and $\tilde{g}^i_t:=g-\hat{g}^i_t$. 
Note that in principle the above equation can be solved by agent $i$ alone, using the local  observations $(x_t^i,y^i_t)\in \Omega^i \times Y$. Each agent  can be controlled to drive the system trajectory $t\mapsto x^i_t$ over the assigned region $\Omega^i$, and the evolution in Equation \ref{eq:strict1} is integrated along the path, without input from other agents.  
On some coarser time scale, these local estimates $\{\hat{g}^i_t\}_{i=1}^N$ constructed by the agents (which eventually  are  reasonable approximations over $\Omega^i$ as  $t\rightarrow \infty$),  are used to define an ideal, infinite dimensional collective estimate $\hat{g}(t)$ that is an approximation over $\Omega$. 
 The ideal,  collective  estimate $\hat{g}(t)$ is defined by 
the expression
$
    \hat{g}(t)=\sum_{i=1}^N \Psi^i \hat{g}^i(t)
$
where the family of functions $\{\Psi^i\}_{i=1}^N$ is a  partition of unity subordinate to the covering $\{\Omega^i\}_{i=1}^N$ of $\Omega$.  The final form of the method, one that yields implementable algorithms, is achieved by introducing approximations of Equation \eqref{eq:strict1a}, the ideal estimate $\hat{g}^i_t$ for each agent, and the collective ideal estimate $\hat{g}_t$.
%
%

Again, so as to place the strictly decentralized method above in context, we briefly summarize the {hybrid decentralized, conensus estimation} method that is studied in part three of this paper. In the hybrid method, each agent $i$ constructs an ideal estimate 
\begin{align*}
    \dot{\hat{g}}^i_t&=\gamma E^*_{x^i_t}\left (y^i_t-E_{x^i_t}\hat{g}^i_t \right ) + \gamma_1 \sum_{j=1}^N E^*_{x^j_t} \mathcal{L}_{ij} \left (y^j_t - E^j_{x^j_t}\hat{g}^j_t \right )
\end{align*}
where $\gamma,\gamma_1$ are positive learning coefficients and  $\mathcal{L}$ is the Laplacian matrix of an undirected graph that defines the communication amongst agents. As in the study of consensus estimation methods in Euclidean space, in the above equation for agent $i$, the rightmost term above can be calculated in terms of a summation over the  neighbors of agent $i$. Propagation of the ideal estimate $\hat{g}^i_t$ {\em requires} sharing certain information among agents continuously in time, at least theoretically.  This form of the hybrid decentralized estimation method is written here to summarize an important feature of the { strictly decentralized} approach.  In contrast to the hybrid method, in the strictly decentralized method  each agent works independently ``as long as possible'' in some sense. Of course, the method in this paper can be used to define a consensus estimation strategy by periodically sharing some approximation of the fused estimates.  To keep our presentation short in this conference paper, we leave all the issues associated with convergence of the consensus methods to part three of the paper. Just as the first paper enables a rather succinct analysis of the rates of convergence in this paper, so too will the rates of convergence in this paper set the stage for the analysis in part three.

\subsection{Our Contributions}
This paper establishes  several new results for  the decentralized adaptive estimation method summarized in the previous subsections. Theorem \ref{th:ideal_strict} gives sufficient conditions for the well-posedness of the DPS that governs the ideal estimate for each agent, as well as a concise statement describing the stability and convergence properties of the ideal DPS defined for each agent. In particular the convergence of the ideal error is based on an appropriate definition of persistence of excitation in the RKHS space.  A method for building realizable approximations of the DPS for each agent is introduced in Equations  \eqref{eq:ct_alpha_evo} and \eqref{eq:dt_alpha_evo}.  Finally, Theorem \ref{th:approx_rate} states conditions under which the approximation error of collective estimate is $O(h_{\Xi_L,\Omega}^p)$ where $h_{\Xi_L,\Omega}$ is the fill distance of the discrete centers $\Xi_L$ in the domain of interest $\Omega$  and $p$ depends on the smoothness of the kernel basis. The paper explains how to choose the kernel from a number of standard examples and determine $p$ in terms of the power function for the kernel in these cases. 
\label{sec:contributions}
\section{Background} \label{sec: bg}

\subsection{Notation and Symbols}
The symbols $\RR,\RR^+$ denote the real numbers and non-negative real numbers respectively, whereas  $\NN$ and $\NN_0$ are used for the positive and non-negative integers. In many instances during a proof, the constant appearing in an inequality does not play a role later. For  this reason we use the expression $a \lesssim b$ to mean that there is a constant $c>0$ that does not depend on $a,b$ such that $a\leq c\cdot b$. An analogous definition holds for $\gtrsim$. By $\mathcal{L}(U,V)$ we denote the collection of bounded, linear operators acting between the normed vector spaces $U$ and $V$. We write $\mathcal{L}(U)$ as shorthand for $\mathcal{L}(U,U)$. For a normed vector space $U$ the spaces of $U-$valued Lebesgue integrable functions over an interval $I\subseteq \RR^+$,  $L^p(I,U)$, have the usual norm 
$\|f\|_{L^p(I,U)}:=\left (\int_I \|f(\tau)\|_{U}d\tau \right)^{1/p}$ for $1\leq p<\infty$, while for $p=\infty$ we have $\|f\|_{L^\infty(I,U)}:=\text{ess sup} \{ \|f(\tau)\|_U\ | \ \tau\in I\}$. 

\subsection{Reproducing Kernel Hilbert Spaces}
This introduction is necessarily brief and focuses on the essentials for this paper.  See \cite{wendland2004scattered,paulsen2016introduction,saitoh2016theory} for a full  account  of the theory of scalar-valued RKHS spaces. All of the RKHS $H_X$   in this paper are induced by a real-valued, admissible kernel function $\knl:X\times X\rightarrow \RR$, so that $H_X$ is a Hilbert space of real-valued functions over $X$. This means that $H_X$ is a Hilbert space in which each of  the evaluation functionals  $E_z:H_X\rightarrow \RR$ for $z\in X$  are linear and  bounded.  When $\knl$ is the kernel that induces the RKHS $H_X$, we  write $\knl_z:=\knl(z,\cdot)$ for the kernel basis function located at $z\in X$. The adjoint $E^*_z:\RR\rightarrow H_X$ of the evaluation functional $E_z$ can be shown to be given by the multiplication by the kernel basis $k_z$: we have $E^*_z \alpha = \knl_z \alpha$ for each $\alpha\in \RR$.  This property  is used throughout the paper in building practical algorithms.   The reproducing property of the native space $H_X$ states that we have $(g,\knl_z)_{H_X}=g(z)$ for every $z\in X$ and $g\in H_X$. 
By definition we have $H_X:=\overline{\text{span}\{\knl_x\ | \ x\in X\}}$, and for any set $A\subset X$ we define the closed subspace $H_A:=\overline{\text{span} \{\knl_x\ | \ x\in A\}}$.  We use $\Pi_A:H_X\rightarrow H_A$ to denote the $H_X$-orthogonal projection onto the subspace $H_A$. 
When we build approximations and define sampling, we will also need to define the spaces $\RH_A:=R_A(H_X)$   that consist of the restriction of functions in $H_X$ to the  subset $A\subset X$, where $R_A$ is the restriction operator.    These are also RKHS spaces in their own right, having  a kernel that is just the restriction of the kernel $\knl$ for $H_X$ to the subset $A$. It can be shown that the subspace $H_A\subset H_X$ and the space of restrictions 
$\RH_A:=R_A(H_X)$ are isomorphic. In fact, we have $\Pi_A=\Ext_AR_A$ where $\Ext_A:\RH_A\rightarrow H_X$ is a  minimal norm extension operator defined as $\Ext_A:=(R_A|_{H_A})^{-1}$. When we  define the inner product on $\RH_A$ as $(f,g)_{\RH_A}:=(\Ext_A f,\Ext_A g)_{H_X}$, this expression coincides with the inner product induced by the restricted kernel.

%
%

 %
 %
\section{Strictly Decentralized Estimation} \label{sec:decentralized}
Recalling the strictly decentralized method in Equation \eqref{eq:strict1a} again here, the ideal  infinite dimensional estimate of agent $i$ for $i=1,\ldots,N$ is defined as the solution of the equation
\begin{align}
\dot{\hat{g}}^i_t&=\gamma E^*_{x^i_t}\left (y^i_t-E_{x^i_t}\hat{g}^i_t \right )\notag \\
&=\gamma E^*_{x^i_t}E_{x^i_t} \tilde{g}^i_t
=\calA^i(t)\tilde{g}^i_t,
\label{eq:strict1}
\end{align}
subject to $\hat{g}^i_t|_0=\hat{g}^i_0$, 
with $\calA^i(t):=\gamma E^*_{x^i_t}E_{x^i_t}$ and $\tilde{g}^i_t:=g-\hat{g}^i_t)$. The associated local error equation is then 
\begin{align}
    \dot{\tilde{g}}^i(t)&=-\calA^i(t) \tilde{g}^i(t), \quad \text{subject to} \quad \tilde{g}^i(0)=\tilde{g}^i_0.
    \label{eq:err_dec} 
\end{align}
Note that in principle the above equation can be solved by agent $i$ alone, using the local  observations $(x_t^i,y^i_t)\in \Omega^i \times Y$. Each agent  can be controlled to drive the system trajectory $t\mapsto x^i_t$ over the assigned region $\Omega^i$, and the evolution in Equation \eqref{eq:strict1} is integrated along the path, without input from other agents.  
On some coarser time scale, these local estimates $\{\hat{g}^i_t\}_{i=1}^N$ constructed by the agents, which are reasonable approximations over $\Omega^i$,  are used to define the ideal, infinite dimensional, collective estimate $\hat{g}(t)$ that is an approximation of $g$ over $\Omega = \bigcup_{i=1}^N\Omega^i$. 

The ideal, collective  estimate $\hat{g}_t$ is defined by 
the expression
\begin{align}
    \hat{g}_t=\sum_{i=1}^N \Psi^i \hat{g}^i_t
    \label{eq:strict3}
\end{align}
where the family of functions $\{\Psi^i\}_{i=1}^N\subset \RH_\Omega$ is a  partition of unity subordinate to the the covering $\{\Omega^i\}_{i=1}^N$.  In other words we have $\Psi^i:\Omega \rightarrow [0,1]$, 
\begin{align}\label{eq:POU_conds}
& \sum_{i=1}^N \Psi^i(x) = 1 &  \forall x\in \Omega, \\
& \text{support}\{\Psi^i \} \subseteq \kappa^i \subseteq \Omega ^i & 1\leq i\leq N. 
\end{align}
We also  assume in our theoretical discussions that the RKHS space $\RH_\Omega$ is invariant under multiplication by $\Psi^i$, so that 
$g\in \RH_\Omega$ implies $\Psi^i \cdot g\in \RH_\Omega$ for $i=1,\ldots,N$. In some practical cases, we relax this assumption, even taking piecewise constant or linear functions for $\Psi^i$. While these latter choices lead to simple and expedient implementations, there can be  some loss of smoothness in the final collective estimate. This will also degrade the convergence results we derive later, but the simplicity of the resulting implementations motivates us to consider these less smooth cases too.  
\subsection{Asymptotic Behavior of the Ideal Estimates} 
In any study of the governing DPS in Equation \eqref{eq: est_eqn}, since the state space is infinite dimensional, it can be a challenge to establish existence and uniqueness of solutions, stability, and convergence of trajectories. 
One advantage of the ideal estimation algorithm described above is that many of its properties can be inferred from the analysis for the centralized case. Indeed,  Equation \eqref{eq:strict1} can be considered to be but  a special case of Equation 
\eqref{eq: est_eqn} when $N=1$. Here we summarize several such results.
\begin{definition}
We say that the trajectory $t\mapsto x^i_t$ persistently excites the set $\Omega^i$ and closed subspace $\calH_{\Omega^i}$ if there are positive constants $\gamma$, $\Delta$ and $T$ such that
\begin{align*}
    \int_{t}^{t+\Delta} \left( E^*_{x^i_t}E_{x^i_t} f, f \right)_{\HsubX} dt \geq \beta \|\Pi_{\Omega^i} f\|_{\HsubX}^2
\end{align*}
for all $f \in \HsubX$ and $t \geq T$.
\end{definition}
\begin{theorem}
\label{th:ideal_strict}
Suppose that $t\mapsto x_t^i\in \Omega^i\subset \Omega \subset \RR^d$ is a continuous trajectory and that there is a constant $\bar{K}>0$ such that $\knl(z,z)\leq \bar{K}$ for all $z\in X$.  Then the following hold:
\begin{enumerate}[label=(\alph*)]
    \item The ideal, strictly local estimation laws represented by  either  Equation \eqref{eq:strict1} or \eqref{eq:err_dec} have  unique solutions $\hat{g}^i,\tilde{g}^i\in \calC([0,T],\HsubX)$ for any finite time $T>0$.
    \item
    The equilibrium at the origin of the error Equation \eqref{eq:err_dec} is stable. In particular,  the trajectory generated for any initial condition is bounded. 
    \item Furthermore, if $ \|x^i_t\|_{\RR^d}\leq c_0$ and $\|\dot{x}^i_t\|_{\RR^d}\leq c_1$ for all $t\in \RR^+$, then the local (ideal) output error $\tilde{y}^i_t:=E_{x^i_t}\tilde{g}^i$ generated by the learning law converges to zero as $t\rightarrow \infty$,
    \begin{align*}
        \lim_{t\rightarrow \infty} \tilde{y}^i_t = 0 \in \RR. 
    \end{align*}
    \item Finally, if the trajectory $t\mapsto x^i_t$ persistently excites the set $\Omega^i$ and closed subspace $\calH_{\Omega^i}$, then 
    \begin{align*}
        \lim_{t\rightarrow \infty} \left \|\Pi_{\Omega^i} \left (\hat{g}^i(t)-g \right ) \right \|_{\HsubX}=0.
    \end{align*}
\end{enumerate}
\end{theorem}
\noindent Statement (a) gives a sufficient condition to ensure the existence of the solution of Equation \eqref{eq:err_dec}. The proof is a standard argument using Banach fixed point theorem. Statement (b) claims the stability of the error system in Equation \eqref{eq:err_dec}, which can be proven using the common quadratic Lyapunov function. Statement (c) is an extension of (b). The additional conditions imposed to the trajectory $x^i_t$ allow us to use Barbalat's lemma to prove the asymptotic stability. Statement (d) claims the convergence of function estimate $\hat{g}^i_t$ if an RKH subspace is persistently excited by the trajectory $x^i_t$.


In Theorem \ref{th:ideal_strict}, we have elected to interpret Equation \eqref{eq: est_eqn} as defining a trajectory in $H_X$. Meanwhile, we can also state this theorem by interpreting Equation \eqref{eq: est_eqn} as defining a trajectory in $\RH_\Omega$ or in $\RH_{\Omega^i}$. In this case, part (a) above guarantees that $\hat{g}^i\in C([0,T],\RH_{\Omega^i})$, for instance, and part (d) concludes that $\|\RH_{\Omega^i}\Pi_{\Omega^i}(\hat{g}^i_t-g)\|_{\RH_{\Omega^i}}$  converges to zero as $t\rightarrow \infty$.

With these properties of the ideal local estimates, the proof of convergence of the 
 collective estimate $\hat{g}(t)$ defined in Equation \eqref{eq:strict3} to the unknown function $g$ is fairly direct. So far we require that the $H_X$ is invariant under multiplication by $\Psi^i$. In the proof of the follwoing corollary, it will likewise be necessary to establish that for each $i$, the multiplication operator $M_{\psi^i}$ that is induced by $\Psi^i$, which is defined by $(M_\psi h)(z):=\Psi^i(z)h(z)$, is in fact a bounded operator, $M_{\Psi^i}\in \calL(H_X)$.  
\begin{corollary}
\label{cor:strict_error}
Suppose the hypotheses of Theorem \ref{th:ideal_strict} hold. Then we have
\begin{align}
    \lim_{t\rightarrow \infty} \|R_\Omega(\hat{g}(t)-g)\|_{\RH_\Omega}=0
\end{align}
for all $g\in H_\Omega$. 
\end{corollary}
\begin{proof}
By definition we have the inequalities
\begin{align*}
    \|R_\Omega(\hat{g}(t)-g)\|_{\RH_\Omega } &:= \left \|\sum_{i=1}^N \Psi^iR_\Omega \hat{g}^i(t) - \sum_{i=1}^N \Psi^i R_\Omega g\right \|_{\RH_\Omega} \\
    &\leq 
    \left \| \sum_{i=1}^N \Psi^iR_\Omega\Pi_{\Omega^i}(\hat{g}^i(t)-g) \right \|_{\RH_\Omega} \\
    & \hspace*{.25in} 
    +
    \left \|
    \sum_{i=1}^N \Psi^iR_\Omega (I-\Pi_{\Omega^i})(\hat{g}^i(t)-g)
    \right \|_{\RH_\Omega}
\end{align*}
But the function $\Psi^i$ is zero outside of $\Omega^i$. By the definition of the projection  $\Pi_{\Omega^i}:
\HsubX\rightarrow \calH_{\Omega^i}$, we also know that $\left ((I-\Pi_{\Omega^i})h\right )(x)=0$ for all $x\in \Omega^i$ and any $h\in \HsubX$. This is true since 
$$
H_{\Omega}^\perp:=\{ h\in H_X \ | \ h(z)=0 \quad \text{ for all } z\in \Omega \}. 
$$
It then follows that 
\begin{align*}
    \left (  \sum_{i=1}^N \Psi^iR_\Omega(I-\Pi_{\Omega^i})(\hat{g}^i(t)-g)\right )(x)=0
\end{align*}
for all $x\in X$. 
We conclude that 
\begin{align*}
    \|R_\Omega \hat{g}(t)-g\|_{\RH_\Omega} & \leq 
     \left \| \sum_{i=1}^N \Psi^iR_\Omega \Pi_{\Omega^i}(\hat{g}-g) \right \|_{\HsubX} \\
     & \leq \sum_{i=1}^N \|M_{\Psi^i}\| \left \|
     R_\Omega \Pi_{\Omega^i} (\hat{g}^i(t)-g) 
     \right \|_{\RH_\Omega},
\end{align*}
if we can show that the multiplication operator $M_{\Psi^i}$ is bounded. From Theorem 5.21 on page 78 of \cite{paulsen2016introduction}, $M_{\Psi^i}$ is a bounded operator on $\RH_\Omega$ if and only if these is a constant $c>0$ such that 
$$
\Psi^i(x)\knl(x,y)\Psi^i(y)\leq c^2 \knl(x,y) \quad \text{ for all } x,y \in \Omega. 
$$
But since $|\Psi(z)|\leq 1$, we can just pick $c=1$. Moreover, the norm $\|M_{\Psi^i}\|$ is just the smallest constant $c$ for which this ineqality holds. This means that in fact we have $\|M_{\Psi^i}\|\leq 1$ for all $1\leq i\leq N$. 
Now the right hand side converges to zero from Theorem \ref{th:ideal_strict}. 
\end{proof}
\subsection{Finite Dimensional Approximations}
\label{sec:strict_dec_FD}
 
The local approximations $\hat{g}^i(t)$ in Equation \eqref{eq:strict1} evolve in an infinite dimensional state space. We must choose some finite dimensional basis for actual implementations. 

\subsubsection{Bases and Approximation}
\label{sec:samples_bases}
The finite dimensional approximations  in the paper make use of a number of well-known properties of the kernel bases $\knl_z$ and the interpolation of functions in a native space  over discrete sets. 
Let us write $\Xi^i_L$ for  a set of $L$ centers  contained in the subdomain $\Omega^i$.  In our approach the centers will ordinarily be taken as samples of the state $x^i(t)$ of agent $i$  collected at some collection of  discrete times $\TT^i_L:=\{ t^i_{L,\ell}\in \RR^+ \ | \ 0\leq \ell \leq L-1 \}$,  
\begin{align*}
\Xi^i_L:=\left \{ \xi_{L,\ell}:=x^i(t^i_{L,\ell}), \ \ t^i_{L,\ell}\in \TT^i_L \right \}
\subset  \Omega^i \subset X. 
\end{align*}
We write 
$
\Xi^i:=\bigcup_{L\geq 1} \Xi^i_L, 
$
for the set of all the centers  over $\Omega^i$, and it is assumed that 
$
\Omega^i \subset \overline{\Xi^i}.
$
We define the collection of $L$  centers  for all of the  subdomains  by  $\Xi_L:=\bigcup_{i=1}^N \Xi^i_L$. 
Asymptotically, the full collection of centers  over all subdomains  is then $\Xi:=\bigcup_{0\leq L<\infty}\Xi_L$.  
We use the sets of centers  $\Xi^i_L$ and $\Xi_L$ to build finite  dimensional spaces $H^i_L$ and $H_L$, respectively. For each agent $i=1:N$, the space of approximants is 
\begin{align*}
    \HsubL^i&:= \text{span}
    \left \{ \knl(\xi^i_{L,\ell},\cdot) \ | \ \xi^i_{L,\ell}\in \Xi^i_L\right \},
\end{align*}
    and we similarly define the global space of approximants as
\begin{align*}
    \HsubL&:=\text{span} 
    \left \{ \knl(\xi^i_{L,\ell},\cdot) \ | \ \xi^i_{L,\ell} \in \Xi_L.
    \right \}
\end{align*}
The operators  $\Pi_L$ and $\Pi_L^i$ are the  $\HsubX$-orthogonal projections onto the two spaces $\HsubL^i$ and $\HsubL$, respectively. 
The finite dimensional spaces $\HsubL^i$ and $\HsubL$ are contained in  the closed subspaces 
\begin{align*}
  {\mathcal{H}}_{\Omega^i} &:= \overline{\text{span}\{\knl(\xi,\cdot)\ | \ \xi \in \Omega^i \} } \subset \HsubX, \\
   {\mathcal{H}}_{\Omega} &: =\overline{ \text{span} \{ \knl(\xi,\cdot) \ | \ \xi \in \Omega \} } \subset \HsubX
\end{align*}
The approximation properties of these finite dimensional spaces are most frequently studied in terms of the fill distance of a finite set $Z$ in another set $A$, which is given by 
\begin{align*}
    h_{Z,A}:=\sup_{x\in A} \min_{z\in Z}d_A(x,z)
\end{align*}
with $d_X$ the metric on $A$. In the discussion that follows, we will have particular need of the fill distances $h_{\Xi^i_L,\Omega^i}$ and $h_{\Xi_L,\Omega}$. The fill distance will be used in conjunction with the power function of the discrete set $Z$, which takes the form 
\begin{align*}
P_{\knl,Z}:=\min_{\alpha\in \RR^J} \left \|\knl(x,\cdot) - \sum_{i=1}^J \alpha_L \knl(z_j,\cdot) \right \|_{R_\Omega(H_X)}
\end{align*}
for a set $ Z:=\{z_1,\ldots,z_J\}$ having $J$ elements. 
\subsection{Local Approximations for Each Agent}
In view of Equation \eqref{eq:strict1a}, each agent $i$ for $i=1,\ldots,N$ constructs the finite dimensional approximation $\hat{g}^i_L(t)\in \calH_L^i$ as the solution of the equation \begin{align}\label{eq:local_est}
    \dot{\hat{g}}^i_L(t)&=
    \gamma \Pi^i_L \calA^i(t)\tilde{g}^i_L(t) & \text{subject to} \quad  \hat{g}^i_L(0)=\hat{g}^i_{L,0},
\end{align}
where $\Pi^i_L:\HsubX\rightarrow \calH_L^i$ is the $\HsubX$-orthogonal projection onto $\calH_L^i$ and $\tilde{g}^i_L(t):=g-\hat{g}_L^i(t)$ is the estimation error, which evolves over time according to 
\begin{align*}
\dot{\tilde{g}}^i_L(t)& =-\gamma \Pi^i_L \calA^i(t) \tilde{g}^i_L(t), & \text{ subject to } \quad \tilde{g}^i_L(0)=\tilde{g}^i_{L,0}. 
\end{align*}
Rates of convergence for the team follow from some fairly standard assumptions regarding the kernel $\knl$ that induces $H_X$. Specifically, it can often be shown that we have 
estimates such as 
\begin{align*}
    P^2_{\knl,Z}(x)\leq F(h_{Z,A}),
\end{align*}
where the function $F(\cdot)$ that defines the upper bound is known or has been characterized. A wide variety of such bounds can be found in Table 11.1 of \cite{wendland2004scattered}. Here we just state the result for Wendland's compactly supported radial basis functions, in which case  it is known that
\begin{align}
P^2_{\knl,\Xi^i_L}\lesssim h^{2s+1}_{\Xi^i_L,\Omega^i}.
\label{eq:power_hp}
\end{align}
We have the following results that establishes a rate of convergence for the fused estimate derived from the local estimates. 
\begin{theorem}
\label{th:approx_rate}
Suppose that $\Omega, \Omega^i$ are bounded, open domains with a sufficiently regular boundary, the kernel $\knl\in C(\Omega\times \Omega)$ is positive definite, and the initial condition $\hat{g}^i_0$ is smooth enough. Further suppose that the power function $P_{\knl,\Xi^i_L}$ of the kernel $\knl$ that induces $H_X$ satisfies a bound of the type in Equation \eqref{eq:power_hp} with $P^2_{\knl,\Omega^i}\lesssim h^p_{\Omega^i,\Omega}$ for some power $p>0$. Then we have 
\begin{align*}
    \|\bar{g}^i_{L,t}\|_{R_{\Omega^i}(H_X)}^2 \leq C
    \left (1+\gamma \|\tilde{y}^i\|^2_{L^2([0,t],\RR)} \right )h^{2p}_{\Xi^i_L,\Omega^i}
\end{align*}
for all $t\in [0,T]$. 
If in addition the multiplication operators $M_{\Psi^i}$ are each bounded, then there is a constant $\bar{C}(t)>0$ and $H_0>0$  such that we have 
\begin{align*}
    \|\bar{g}_{L,t}\|_{R_\Omega(H_X)} \lesssim \bar{C}(t) \max_{i=1,\ldots,N} h^{2p}_{\Xi^i_L,\Omega^i} 
\end{align*}
for each $t\in [0,T]$
\end{theorem}
\begin{proof}
 The rate of convergence of the finite dimensional approximations $\hat{g}^i_L(t)$ to the ideal estimate $\hat{g}^i(t)$ can be bounded by applying Theorem \ref{th:error1} and Corollary \ref{cor:cor1} for the trivial case that we have $N=1$ and $\Omega=\Omega^i$. We obtain
$$
\|\bar{g}^i_L(t)\|^2_{\HsubX}\leq C \left ( 1+\gamma \|\tilde{y}^i\|^2_{L^2([0,t],
\RR)}\right ) h^{2p}_{\Xi^i_L,\Omega^i}. 
$$
\end{proof}
\subsection{Discrete-time Evolution Law} Given that $\ghat_L^i(t)\in \HsubL^i$, we can express the estimator for each agent $i=1{:}N$ as
\begin{align}\label{eq:lin_exp}
     \ghat_{L,t}^i(\cdot) = \sum_{\ell = 1}^L\alpha_{t,\ell}\calk({\xi_{\ell},\cdot)}.
\end{align}
It follows from the derivation in Section 3.2.3 of \cite{centralized_paper}, that the coefficients $\alpha^i_t \in \bbR^L$ evolve over time according to
\begin{align}\label{eq:ct_alpha_evo}
    \dot{\alpha}_t^i = \gamma \KK^{-1}(\Xi_L^i,\Xi_L^i)^{-1}\KK\left(\Xi_L^i,x_t^i)(y_t^i-\KK(x_t^i,\Xi_L^i)\alpha^i_t\right)
\end{align}
where the kernel matrices are defined element-wise via $[\KK(\Xi_L^i,\Xi_L^i)]_{m,n} = \calk(\xi_m,\xi_n)$ for $m,n=1{:}L$ and $[\KK(\Xi_L^i,x_t^i)]_{m} = \calk(\xi_m,x_t^i)$ for $m=1{:}L$. In practice, the evolution law \eqref{eq:ct_alpha_evo} is carried out using a discrete time multistep integrator such that for each time instance $k = 1,2,...$ the evolution takes the form
\begin{align}\label{eq:dt_alpha_evo}
    \alpha_{k+1}^i = \alpha_k^i +h \gamma \sum_{s=1}^q a_s \calk(x_{k-s}^i,\cdot)(y_{k-s}^i-\ghat_{k-s}^i(x_{k-s}^i)).
\end{align}
The order $q$ and associated coefficients $\mathbf{a} = \{a_s\}_{s=1}^q$ are determined by the particular multistep method used; see \cite[Ch. 2.24]{butcher2016numerical} for details. For simplicity, we consider a fixed time step $h$, although a time-varying step size could also be used.
\subsection{Basis Enrichment}
When agent $i$ obtains a new input-sample pair $(x,y)$, it must decide whether to add $x$ to the set of basis centers $\Xi^i_L$. To do so, we introduce a notion of novelty and only sufficiently novel inputs are added to the $\Xi^i_L$. As done in \cite{csato2002sparse,gao2019gaussian,keplerapproach} we make the notion of novelty precise in terms of the squared norm of the residual of $k_x$ projected onto $\HsubL^i$ given by
\begin{align} \label{eq:nov_metric}
    \begin{split}
        \epsilon^2 &= \Vert \calk_x - \Pi_L^i \calk_x \Vert_{\HsubX}^2 \\
                 &= \calk(x,x) - \KK(x,\Xi_L^i)\KK(\Xi_L^i,\Xi_L^i)^{-1}\KK(\Xi_L^i,x). 
    \end{split}
\end{align}
Thus, if $\epsilon$ is greater than some user-defined $\bar{\epsilon}>0$, then $x$ is considered sufficiently novel and added to $\Xi_L^i$. In practice,  $\bar{\epsilon}$ can be tuned to ensure the Gram matrix $\KK(\Xi_L^i,\Xi_L^i)$ is well-conditioned as is done in \cite{csato2002sparse}. 

Suppose at time $k := t_k$ the set of basis centers $\Xi_{L_k}^i$ contains $L_k = L-1$ elements for any arbitrary agent $i=1{:}N$. Whenever a new element $\xi \in X$ is added such that the set of centers becomes $\Xi_{L_k}^i \cup {\xi}$. Associated with the augmented set of centers, is a new set of prediction coefficients of size $L$. We initialize the new coefficients so as to interpolate the spatial field at the locations contained in the set of basis centers $\Xi_{L_k}^i$, i.e.
\begin{align*}
    \alpha^i_k = \KK(\Xi_{L_k}^i,\Xi_{L_k}^i)^{-1}\mathbf{g}^i,
\end{align*}
where $\mathbf{g}^i:=g(\Xi_{L_k}^i)\in \bbR^{L}$. The evolution of the prediction coefficients and the process of basis enrichment is summarized in Algorithm \ref{alg: training}.

\IncMargin{1em}
\begin{algorithm}
\SetKwInOut{Input}{input}\SetKwInOut{Output}{output}
\SetKwInOut{Input}{input}\SetKwInOut{Output}{output}
\Input{Estimator kernel $\calk$, kernel hyp. $\theta$, novelty threshold $\bar{\epsilon}>0$}
\Output{Prediction coefficients $\{\alpha^i_k\}_{i=1}^N$}
\BlankLine
\For{i=1:N}{
    Collect initial samples: $X_0^i,Y_0^i$ \;
    Init. centers \& sample set: $\Xi_{L_0}^i = X_0^i$; $\mathbf{g}^i = Y_0^i$\;
    Init. coefs.: $\alpha^i_0 = \KK(\Xi_{L_0}^i,\Xi_{L_0}^i)^{-1}\mathbf{g}^i$\;
}
\While{$\mathrm{collecting \ samples}$}{
    $k + 1 \rightarrow k$\;
    \For{ i = 1:N}{
        Sample spatial field: $(x_k^i, y_k^i)$ \;
        Coefficient update \eqref{eq:dt_alpha_evo}: $\alpha_k^i \rightarrow \alpha_{k+1}^i$ \;
        Calc. novelty of $x_k^i$ wrt $\Xi^i_{L_k}$ \eqref{eq:nov_metric}: $\epsilon$ \;
        \If{$\epsilon > \bar{\epsilon}$}{
            Update centers: $ \{\Xi_{L_k}^i \cup x_k^i\} \rightarrow \Xi_{L_k}^i$\; 
            Update samples: $\{\mathbf{g}^i \cup y_k^i \}\rightarrow \mathbf{g}^i$ \;
            Init. coefs.: $\alpha^i_k = \KK(\Xi_{L_k}^i,\Xi_{L_k}^i)^{-1}\mathbf{g}^i$ \;
        }
    }
}
\caption{Evolution of the prediction coefficients for the decentralized estimator.}\label{alg: training}
\end{algorithm}
\subsection{Fused Estimate}
Suppose at time $k := t_k$, each agent $i=1{:}N$ has its own set of basis centers $\Xi_{L_k^i}$ and associated coefficients $\{\alpha^i_{k,\ell}\}_{\ell = 1}^{L_k^i}$. In view of the ideal fused estimate \eqref{eq:strict3}, we define a computationally tractable approximation such that for any $x \in \Omega$
\begin{align}\label{eq:tract_est}
    \ghat_{k,L}(x) = \sum_{i=1}^N\psi^i(x) \sum_{\ell = 1}^{L_k^i}\alpha_{k,\ell}^i\calk(\xi_\ell^i,x),
\end{align}
where $\{\psi^i\}^N_{i=1}$ is the partition of unity.

In order for a single entity, whether that be an agent of the team or an independent processor, to compute the fused estimate \eqref{eq:tract_est}, all basis centers $\{\Xi_{L_k^i}\}_{i=1}^N$ and prediction coefficients $\{\alpha^i_k\}_{i = 1}^{N}$ must be communicated to the computing entity. Note that communication can potentially be reduced if the entity is only interested in making predictions in a particular subset $E \subset \Omega$ of the domain. For $x \in E$, the estimator \eqref{eq:tract_est} can be expressed as
\begin{align}\label{eq:red_tract_est}
    \ghat_{k,L}(x) = \sum_{j\in \mathcal{N}(x)} \psi^j(x) \sum_{\ell = 1}^{L_k^j}\alpha_{k,\ell}^j\calk(\xi_\ell^j,x).
\end{align}
where $\mathcal{N}(x) := \{j\in \{1,...,N\}\ |\ \psi^j(x) \neq 0\}$. Thus, a reduction in communication is achieved when $|\mathcal{N}(x)| < N$. Indeed, communication can be trivially eliminated by having each agent $i=1{:}N$ simply use its local estimator $\ghat^i_k$ over its assigned domain $\Omega^i$. However, this reduction in communication comes at the expense of a less-informed estimate.

\section{Numerical Illustration}
To illustrate the behavior of the proposed decentralized estimation method, we consider a spatial field $g$ that is an element of the RKHS $\HsubX$ induced by the 5/2 Matern Kernel \cite[Ch. 4.2]{williams2006gaussian}
\begin{align}\label{eq:matern}
    \calk(x,x') = \sigma^2\left( 1+\frac{\sqrt{5}r}{\ell}+\frac{5r^2}{3\ell^2} \right) \exp\left\{-\frac{\sqrt{5}}{\ell}\right\},
\end{align}
where $r := \Vert x-x' \Vert_{\bbR^d}$ and the set of hyperparameters $\theta = [\sigma,\ell]$ consists of the scale factor $\sigma$ and length-scale parameter $\ell$. We fix the kernel hyperparameters of $g$ as $\theta_g = [1,1]$, and construct $g$ by first discretizing the domain into a grid $\Xi_g \subset \Omega$ that consists of $N_g$ inputs. We then obtain a realization $\mathbf{g}\in\bbR^{N_g}$ of a zero-mean Gaussian random variable whose covariance matrix is $\KK(\Xi_g,\Xi_g)$. With $\mathbf{g}\in\bbR^{N_g}$ fixed, the spatial field can be expressed as
\begin{align*}
    g(\cdot) &= \sum_{i=1}^{N_g} \alpha_{i} \calk(\xi_i,\cdot),  & \alpha &= \KK(\Xi_g, \Xi_g)^{-1}\mathbf{g}.
\end{align*}
We hold $g$ fixed across all experiments and we consider decentralized estimation for the case of $N=4$ agents. The subdomain assignments $\{\Omega^i\}_{i=1}^4$ along with $g$ are depicted in Figure \ref{fig: gAndOmega}.

\begin{figure}
  \centering
  \includegraphics[width = \linewidth]{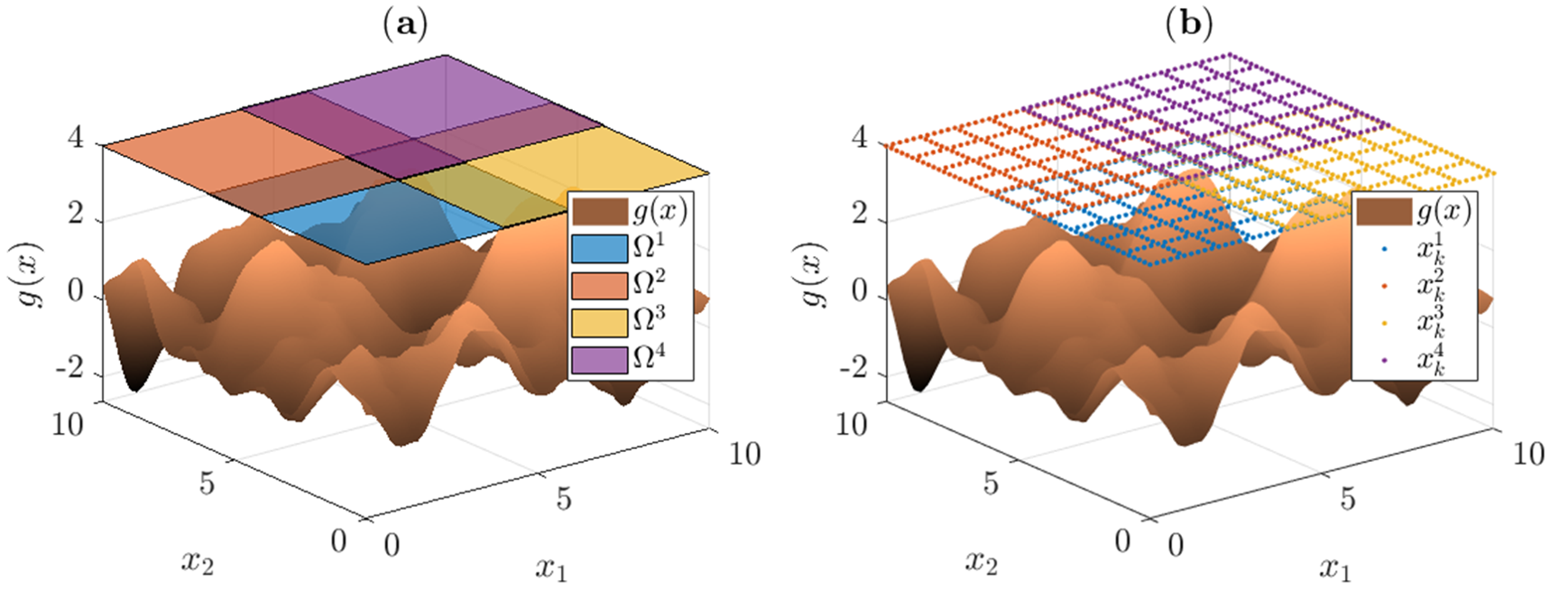}
  \caption{Spatial field used in the numerical experiments overlaid with (a) the subdomain assignments $\{\Omega^i\}_{i=1}^N$ for each of the $N=4$ agents and (b) the agent sampling locations $\{x_k^i\}_{i=1}^N$ when the grid resolution is 1.}
\label{fig: gAndOmega}
\end{figure}

\begin{figure*}
  \centering
  \includegraphics[width = \linewidth]{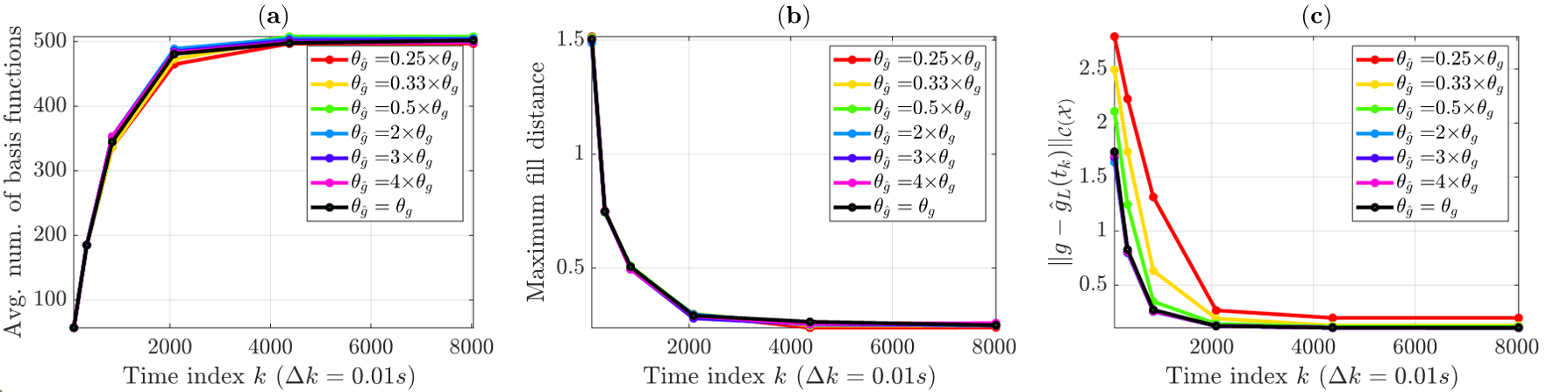}
  \caption{Depicted in subplot: (a) Average number of basis functions for each agent: $\frac{1}{N}\sum_{i=1}^N L_k^i$ (b) maximum fill distance across all agents: $\max_{i=1:N} h_{\Xi_{L_k}^i,\Omega^i}$ (c) Estiamtor error measure: $\Vert g-\ghat_L(t_k) \Vert_{\mathcal{C(X)}}$. The black traces correspond to the ideal case in which the estimator kernel hyperparamters $\theta_{\ghat}$ are equal to the spatial field kernel hyperparamters $\theta_g$, while the colored traces reflect a more practical setting in which $\theta_{\ghat}\neq \theta_g$}
\label{fig: seq_results}
\end{figure*}

To simulate Algorithm \ref{alg: training}, we fix the kernel $\calk$ of each estimator to coincide with the kernel \eqref{eq:matern} of $g$. In practice, the kernel hyperparameters $\theta_g$ are seldom known. We perform a simulation for the ideal case in which the kernel hyperparameters of each estimator $\theta_{\ghat}$ are equal to $\theta_g$, and we perform simulations in which $\theta_{\ghat} = c\theta_g$ for each $c\in\{\frac{1}{4},\frac{1}{3},\frac{1}{2},2,3,4\}$ to emulate a more realistic scenario in which $\theta_{\ghat} \neq \theta_g$. For each set of estimator hyperparameters $\theta_{\ghat}$, the novelty threshold $\bar{\epsilon}$ is tuned so that set of kernel centers each agent contains approximately 500 elements.

The sampling trajectory of each agent $i=1{:}N$ is constructed from a sequence of "lawnmower" grid-like paths with each path in the sequence traversing the entire subdomain. The first path in the sequence has a relatively coarse grid resolution, and the grid resolution of subsequent paths are progressively refined so that we can observe the behavior of the estimator $\ghat_L(t)$ as the sampling trajectory of each agent becomes dense in their respective subdomain and the fill distance tends to zero. See Figure \ref{fig: gAndOmega} (b) for the sampling trajectory associated with the lawnmower path whose grid resolution is 1. 

 We assume that when agents are simultaneously in regions where their respective subdomains overlap, the agents exchange their basis centers and prediction coefficients. For example, if at time $t$ both $x^1_{t}$ and $x_t^2$ are in $\Omega^1\cap\Omega^2$ then agents 1 and 2 share their most recent sets of basis centers and prediction coefficients. 
 
 We use $\eqref{eq:tract_est}$ to form the fused estimator $\ghat_L(t)$. Note that in general, the partition of unity is tailored to the application at hand subject to the conditions introduced in \eqref{eq:POU_conds}. We define the partition of unity so as to effectively average the estimates in a particular region. For example, $\psi^1(x)$ is defined such that
\begin{align*}
    \psi^1(x) = \left\{ \begin{matrix} 1, &  x\in \Omega^1 \setminus  \{\Omega^2 \cup  \Omega^3 \cup \Omega^4 \}\\
                                      1/2, &  x\in \{\Omega^1 \cap \Omega^2\} \setminus \{\Omega^3 \cup \Omega^4 \} \\
                                      1/2, &  x\in \{\Omega^1 \cap \Omega^3\} \setminus \{\Omega^2 \cup \Omega^4 \} \\
                                      1/4, &   x\in  \Omega^1 \cap \Omega^2 \cap  \Omega^3 \cap \Omega^4 \\
                                      0,   & \textrm{otherwise}
    \end{matrix}, \right.
\end{align*}
and the remaining functions $\{\psi^i\}_{i=2}^4$ are defined analogously.

In Figure \ref{fig: seq_results}, we report the time evolutions for 
\begin{enumerate}[label=(\alph*)]
\item the average number of kernel basis centers across all agents, i.e. $\frac{1}{N} \sum_{i=1}^N \vert \Xi_L^i \vert$, 
\item the largest fill distance of all the agents, i.e. $\max_{i=1:N}(h_{\Xi_{L}^i,\Omega^i})$, and
\item the estimator error measure $\Vert g - \ghat_L(t_k) \Vert_{\mathcal{C(X)}}$.
\end{enumerate}
Observe that as more basis centers are added and the fill distance tends to zero, the estimator error measure tends to zero even in instances where $\theta_{\ghat} \neq \theta_g$. In Figure \ref{fig:err_surf}, we overlay the final estimator error $\vert \ghat_L(t) - g \vert$ surface with the locations of the basis centers associated with each agent for the case when $\theta_{\ghat} = \theta_g$. Qualitatively speaking, we see that the estimator error is small when the set of all basis centers is dense in the domain $\Omega$.
\begin{figure}
  \centering
  \includegraphics[width = 2.75in]{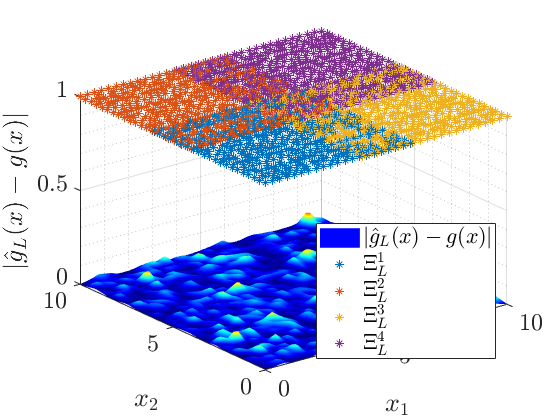}
  \caption{Estimator error surface after all agents have completed their sampling trajectories for the ideal case in which the estimator kernel hyperparameters coincide with the kernel hyperparameters of the underlying spatial field, i.e. $\theta_{\ghat}=\theta_{g}$.}
\label{fig:err_surf}
\end{figure}

\section{Conclusions} \label{sec:conclusions}
We have established that the proposed decentralized estimation framework is well-posed and has convergence rates that rely on (i) a notion of persistence excitation in reproducing kernel Hilbert spaces and (ii) the fill distance of sampling locations used to construct the finite-dimensional approximation. Given that each agent $i=1{:}N$ only uses its own measurements to update its estimator $\ghat^i_t$, much of the analysis the proposed decentralized estimation scheme follows from the centralized scheme for the case when $N=1$ and $\Omega^i=\Omega$. In contrast to the centralized setting, the strictly decentralized setting requires that we fuse the collection of local estimators $\{\ghat_t^i\}_{i=1}^N$ into a global estimator $\ghat_t$. By using a partition of unity to fuse the local estimators, we are able to establish that the estimator error $\Vert \ghat_t-g \Vert_{H_{\Omega}}$ converges to zero. The novel error analysis presented here for the strictly decentralized setting lays the ground work for part III of the family of papers in which we consider a hybrid decentralized estimation scheme in which each agent uses its own measurements and measurements of neighboring agents to update its estimator.

\bibliographystyle{./IEEEtran}
\bibliography{./IEEEabrv,./references}

\begin{thebibliography}{1}
\providecommand{\url}[1]{#1}
\csname url@rmstyle\endcsname
\providecommand{\newblock}{\relax}
\providecommand{\bibinfo}[2]{#2}
\providecommand\BIBentrySTDinterwordspacing{\spaceskip=0pt\relax}
\providecommand\BIBentryALTinterwordstretchfactor{4}
\providecommand\BIBentryALTinterwordspacing{\spaceskip=\fontdimen2\font plus
\BIBentryALTinterwordstretchfactor\fontdimen3\font minus
  \fontdimen4\font\relax}
\providecommand\BIBforeignlanguage[2]{{%
\expandafter\ifx\csname l@#1\endcsname\relax
\typeout{** WARNING: IEEEtran.bst: No hyphenation pattern has been}%
\typeout{** loaded for the language `#1'. Using the pattern for}%
\typeout{** the default language instead.}%
\else
\language=\csname l@#1\endcsname
\fi
#2}}

\bibitem{wendland2004scattered}
H.~Wendland, \emph{Scattered data approximation}.\hskip 1em plus 0.5em minus
  0.4em\relax Cambridge university press, 2004, vol.~17.

\bibitem{paulsen2016introduction}
V.~I. Paulsen and M.~Raghupathi, \emph{An introduction to the theory of
  reproducing kernel Hilbert spaces}.\hskip 1em plus 0.5em minus 0.4em\relax
  Cambridge University Press, 2016, vol. 152.

\bibitem{saitoh2016theory}
S.~Saitoh and Y.~Sawano, \emph{Theory of reproducing kernels and
  applications}.\hskip 1em plus 0.5em minus 0.4em\relax Springer, 2016.

\bibitem{butcher2016numerical}
J.~Butcher, \emph{Numerical Methods for Ordinary Differential Equations}.\hskip
  1em plus 0.5em minus 0.4em\relax John Wiley \& Sons, 2016.

\bibitem{csato2002sparse}
L.~Csat{\'o} and M.~Opper, ``Sparse on-line gaussian processes,'' \emph{Neural
  computation}, vol.~14, no.~3, pp. 641--668, 2002.

\bibitem{gao2019gaussian}
T.~Gao, S.~Z. Kovalsky, and I.~Daubechies, ``Gaussian process landmarking on
  manifolds,'' \emph{SIAM Journal on Mathematics of Data Science}, vol.~1,
  no.~1, pp. 208--236, 2019.

\bibitem{keplerapproach}
M.~E. Kepler and D.~J. Stilwell, ``An approach to reduce communication for
  multi-agent mapping applications,'' in \emph{2020 IEEE/RSJ International
  Conference on Intelligent Robots and Systems (IROS)}.\hskip 1em plus 0.5em
  minus 0.4em\relax IEEE, 2020, pp. 4814--4820.

\bibitem{williams2006gaussian}
C.~Williams and C.~E. Rasmussen, ``Gaussian processes for machine learning,
  vol. 2,'' \emph{MIT press Cambridge, MA}, vol. 302, p. 303, 2006.

\end{thebibliography}

\end{document}